\documentclass[11pt]{article}%
\usepackage{algorithm}
\usepackage{algpseudocode}
\usepackage{amssymb}
\usepackage{amsfonts}
\usepackage{amsmath}
\usepackage{graphicx}%
\usepackage[caption=false]{subfig}
\providecommand{\U}[1]{\protect\rule{.1in}{.1in}}
\usepackage{verbatim}
\usepackage{amsmath,amssymb,amsthm}

\usepackage{url}
\usepackage{hyperref}

\newcommand{\quotes}[1]{``#1''}

\usepackage{graphicx}
\usepackage{array}
\usepackage{booktabs}

\newtheorem{theorem}{Theorem}

\newtheorem{definition}[theorem]{Definition}

\newtheorem{problem}[theorem]{Problem}

\renewcommand{\O}{\mathcal{O}}

\usepackage{tikz}
\usetikzlibrary{shapes,positioning}

\tikzset{
	mystyle/.style={
		circle,
		align=center,
		draw=black,
		fill=cyan!30!white,
		minimum size=6mm,
		line width=0.4mm, 
	}
}

\begin{document}

\title{On the non-submodularity of the problem of adding links to minimize the effective graph resistance}
\author{Massimo A. Achterberg* \ and Robert E. Kooij*\textsuperscript{,$\dagger$}}
\date{* Faculty of Electrical Engineering, Mathematics and Computer Science, Delft University of Technology, P.O. Box 5031, 2600 GA Delft, The Netherlands \\
	\textsuperscript{$\dagger$} Unit ICT, Strategy \& Policy, TNO (Netherlands Organisation for Applied Scientific Research), P.O. Box 96800, 2509 JE, Den Haag, The Netherlands \\
	email: \{M.A.Achterberg, R.E.Kooij\}@tudelft.nl \\
	\today}
\maketitle

\begin{abstract}
We consider the optimisation problem of adding $k$ links to a given network, such that the resulting effective graph resistance is as small as possible. The problem was recently proven to be NP-hard, such that optimal solutions obtained with brute-force methods require exponentially many computation steps and thus are infeasible for any graph of realistic size. Therefore, it is common in such cases to use a simple greedy algorithm to obtain an approximation of the optimal solution. It is known that if the considered problem is submodular, the quality of the greedy solution can be guaranteed. However, it is known that the optimisation problem we are facing, is not submodular. For such cases one can use the notion of generalized submodularity, which is captured by the submodularity ratio $\gamma$. A performance bound, which is a function of $\gamma$, also exists in case of generalized submodularity. In this paper we give an example of a family of graphs where the submodularity ratio approaches zero, implying that the solution quality of the greedy algorithm cannot be guaranteed. Furthermore, we show that the greedy algorithm does not always yield the optimal solution and demonstrate that even for a small graph with 10 nodes, the ratio between the optimal and the greedy solution can be as small as~0.878.

\end{abstract}

\section{Introduction}
Many network metrics have been utilised to quantify the robustness of a network, see for instance~\cite{ps18,jose,fiedler,schneider,hale}. Freitas \emph{et al}.\ \cite{freitas} classify robustness metrics into three types: 
metrics based on structural properties, such as edge connectivity or diameter; metrics based on the spectrum of the adjacency matrix, such as the spectral radius or spectral gap; and metrics based on the spectrum of the Laplacian matrix, for instance the algebraic connectivity and the effective graph resistance.
The robustness is considered optimal for a complete network and minimal for disconnected networks. In this paper we consider the $k$-Graph Robustness Improvement Problem ($k$-GRIP) \cite{predari2022graphresistance}, in which one has to decide where $k$ links are to be added to a given network $G$, such that the robustness is optimised. The set of placeable positions is not necessarily all currently non-existing links -- there may be additional constraints. We call the set of placeable links $V$ and denote the robustness measure by~$f$. From here onward, we assume that the set of placeable links $V$ equals all non-existing links, unless otherwise specified. 

For many choices of the robustness metric $f$, $k$-GRIP is known to be NP-hard. To overcome the difficulty of finding the optimal solution (often only possible using a brute-force algorithm), we apply a simple greedy algorithm. Out of all placeable links~$V$, at each step the greedy algorithm selects a single link to add. This procedure is repeated until $k$ links are added. The greedy algorithm often performs well in practice, but the solution quality in general cannot be guaranteed.

The notion of submodularity was introduced by Nemhauser \emph{et al}.\ \cite{nemhauser1978subdmodularity} as a tool to guarantee the performance of the greedy algorithm. Submodularity implies that adding one element to a large set has relatively small influence, whereas it has a stronger influence on a small set. Nemhauser \emph{et al}.\ \cite{nemhauser1978subdmodularity} proved that if the robustness metric satisfies the submodularity condition, the corresponding optimisation problem can be solved with a polynomial-time greedy algorithm, whose performance is at least a factor $\left(1-\frac{1}{e}\right)$-close to the optimal solution, and moreover, there does not exist any better algorithm with the same complexity. A recent overview of submodularity is provided by \cite{clark2017submodularity}. The concept of submodularity was generalised by Das and Kempe \cite{daskempe2011submodularityratio} by introducing the submodularity ratio~$\gamma$ and was further generalised with the concept of  curvature~$\alpha$ in \cite{summers2019curvature,bian2017curvature} and for nonmonotone metrics \cite{santiago2020submodularity}.

In this paper we consider the effective graph resistance \cite{ellens2011graphresistance} as the robustness metric. The effective graph resistance not only covers the shortest path between any pair of nodes, but incorporates all paths between any two nodes. It has been shown in \cite{summers2017correction} that the problem of minimizing the effective graph resistance upon the addition of $k$ links is non-submodular. In this paper we give an example of a family of graphs where the submodularity ratio approaches zero, implying that, even by using the concept of generalized submodularity, the solution quality of the greedy algorithm cannot be guaranteed. 

The remainder of the paper is as follows. We first show related work in Section~\ref{sec_related_work}. Section~\ref{sec_background} formally introduces the $k$-GRIP optimisation problem. We proceed by providing a counterexample for generalised submodularity in Section~\ref{sec_counterexample} for $k$-GRIP with the effective graph resistance. We compare the greedy algorithm with the brute-force method for many small graphs in Section~\ref{sec_greedy} and show that the greedy algorithm does not always provide the optimal solution. Finally, we conclude in Section~\ref{sec_conclusion}.

\section{Related work}\label{sec_related_work}
Several researchers investigated $k$-GRIP for specific robustness metrics. For instance, \cite{wang2008algebraic} considered $1$-GRIP, with as robustness metric the algebraic connectivity, i.e.\ the second-smallest eigenvalue of the Laplacian matrix $L$. They suggest several strategies, based upon topological and spectral properties of the graph, to decide which single link to add to the network, in order to increase the algebraic connectivity as much as possible. The algebraic connectivity for $k$-GRIP was considered by \cite[Chapter 8]{ZhidongThesis}. Under some light conditions, lower bounds for the quality of the greedy solution were obtained. It might be argued that the algebraic connectivity is not a proper robustness metric, because there are examples where adding a link to a graph, does not change the algebraic connectivity, see \cite{natural_connectivity}. The NP-hardness of $k$-GRIP for the algebraic connectivity was proved in \cite{AC_NP}. A nice overview of $k$-GRIP for the algebraic connectivity is presented in \cite{li2018algebraicconnectivity}, see the references [5--16] therein. 

Shan \emph{et al}.\ \cite{shan2018noderesistance} considered the node resistance as robustness metric, which is the sum of the effective resistances from one source node $v$ to all other nodes. They assume $V$ is the set of non-existing links from the source node $v$; not all possible non-existing links. In that case, the node resistance is shown to be submodular. Papagelis~\cite{papagelis2015averagepathlength} shows that $k$-GRIP with the average shortest path length as a robustness metric does not satisfy the submodularity constraint, but accurate greedy solutions can be obtained. Van~Mieghem \emph{et al}.\ \cite{vanmieghem2011spectralradius} consider a link removal problem with the spectral radius (largest eigenvalue of the adjacency matrix) as a robustness metric and prove this problem is NP-hard. Baras and Hovareshti \cite{baras2009spanningtrees} consider the problem of adding $k$ links to a given network, such that the number of spanning trees in the graph is maximised.

We investigate the effective graph resistance (also known as the Kirchhoff index), which was proposed as robustness metric in \cite{ellens2011graphresistance}. The NP-hardness of $k$-GRIP for the effective graph resistance was proven in \cite{kooij2023NPhard}. Summers \emph{et al}.\ \cite{summers2015graphresistance} attempted to prove that $k$-GRIP with the effective graph resistance is submodular. Later, they corrected their own statement in an online document \cite{summers2017correction}, showing a counterexample for submodularity. Nevertheless, the greedy algorithm appears to yield near-optimal solutions. Wang \emph{et al}.\ \cite{wang2014graphresistance} considered adding a single link and derived bounds for the quality of the greedy algorithm. They additionally investigated different strategies to find the optimal link to add. Pizzuti and Socievole~\cite{pizzuti2018graphresistance} introduce a genetic algorithm as a heuristic to find the best link to add. Clemente and Cornaro~\cite{clemente2020graphresistance} derived bounds for the effective graph resistance after adding/removing one or multiple links. Ghosh \emph{et al}.\ \cite{ghosh2008graphresistance} considered the case of weighted links, under a fixed allocation budget, for which an efficient (polynomial-time) algorithm is provided. Etesami~\cite{etesami2021resistance} and Chan \emph{et al}.\ \cite{chan2022resistance} consider maximising and minimizing the effective graph resistance between source node $s$ and target node $t$ under a fixed allocation budget, respectively.

Results for $k$-GRIP may, besides the considered robustness metric, also depend on the optimisation problem itself. For example, \cite{clark2011graphresistance,clark2014graphresistance,clark2017graphresistance} consider a node-selection optimisation problem with the effective graph resistance as robustness metric, whereas $k$-GRIP considers link addition with the effective graph resistance, which is fundamentally different. Even though the objective function is the same (minimising the effective graph resistance), the problem constraints are very different. In their case, submodularity of the effective graph resistance holds, whereas in our case, we prove that the considered problem does not even satisfy the condition of generalized submodularity.

\section{Background}\label{sec_background}
In this paper we consider undirected, connected simple graphs $G=(V,E)$ without self-loops. Here $V$ denotes the set of $N$ vertices, while $E$ is the set of $L$ links connecting vertex pairs of $V$. The notation $i \sim j$ indication that nodes $i$ and $j$ are adjacent in $G$. We let $G^c = (V,E^c)$ denote the complementary graph of $G$, where $E^c = \{(u,v) | u, v \in V, u \neq v, (u,v) \not\in E\}$. 
The adjacency matrix $A$ of $G$ is an $N \times N$ symmetric matrix with elements $a_{ij}$ that are either 1 or 0 depending on whether there is a
link between nodes $i$ and $j$ or not. The Laplacian matrix $Q$ of $G$ is an $N \times N$ symmetric matrix $Q=\Delta - A$, where $\Delta = diag(d_i)$ is the $N \times N$ diagonal degree matrix with the elements $d_i = \sum_{j=1}^N a_{ij}$.  The eigenvalues of $Q$ are all real and non-negative and can be ordered as $0=\mu_1 \leq \mu_2 \leq \cdots \leq \mu_N$.

Interpreting the graph $G$ as an electrical network whose edges are resistors of $1\Omega$, the effective resistance~$\omega_{ij}$ between node $i$ and $j$ can be computed based on Kirchoff's circuit laws, where it assumed that a unit current is injected into $G$ at $i$ and extracted at $j$. Then the \emph{effective graph resistance}~$R_G$, also known as the \emph{Kirchhoff index}, is defined as the sum of the effective resistances over all node pairs \cite{kleinrandic1993resistance}:

\begin{equation}\label{eq_graph_res_def}
	R_G(G) = \sum_{1 \leq i < j \leq N} \omega_{ij}.
\end{equation}

The effective graph resistance $R_G$ can also be related to eigenvalues of the Laplacian matrix $L$ in the following way \cite{ellens2011graphresistance}
\begin{equation}\label{eq_graph_res_laplacian}
	R_G = N \sum_{i=2}^N \frac{1}{\mu_i}
\end{equation}
where $\mu_i$ denotes the $i^{th}$ eigenvalue of the Laplacian matrix $L$, where the eigenvalues are ordered from small to large. 

We can now formally formulate the optimization problem we want to address.

\begin{problem}[$k$-GRIP for the effective graph resistance]
	Given an undirected, connected, simple graph $G = (V,E)$ and a non-negative integer $k$, find a subset $B \subseteq E^c$ of size $|B| = k$ which minimizes the effective graph resistance $R_G(H)$ for the graph $H = (V, E \cup B)$.
	\label{problem_RG}
\end{problem}

In order to formulate the performance bound for the greedy algorithm given by \cite{nemhauser1978subdmodularity}, we first need two definitions.
\begin{definition}\label{def_monotone}
	A function $f$ is monotonically increasing if and only if $f(S) \leq f(T)$ for all $S \subseteq T$.
\end{definition}

The notion of submodularity is defined as follows.
\begin{definition}[\cite{nemhauser1978subdmodularity}]\label{def_submodularity}
	A function $f$ on a set $W$ is called \textbf{submodular} if
	\begin{equation*}
		f(S \cup \{v \}) - f(S) \geq f(R \cup \{v\}) - f(R)
	\end{equation*}
	for all $S \subseteq R \subset W$ and all $v \in W \backslash R$.
\end{definition}

According to~\cite{nemhauser1978subdmodularity}, if the function $f$ is monotonically increasing and submodular, then the obtained greedy solution is at least $(1-\frac{1}{e})$-close to the optimal solution.  
However, the effective graph resistance is known to be monotonically decreasing upon the addition of links \cite[Thm 2.7]{ellens2011graphresistance}. Therefore, we propose a linear transformation of the effective graph resistance such that the resulting function is monotonically increasing, and that for connected graphs its value is bounded between 0 and 1. According to \cite{ghosh2008graphresistance} for connected graphs the maximum value for the effective graph resistance is $\binom{N-1}{3}$, which is obtained for the path graph $P_N$. The minimum value for the effective graph resistance, which is $N-1$, occurs for the complete graph~$K_N$. We define the normalized effective graph resistance $r_G$ as
\begin{equation}\label{eq_graph_resistance_scaled}
	r_G = \frac{\binom{N-1}{3} - R_G}{\binom{N-1}{3} - (N-1)}
\end{equation}
which is bounded between 0 and 1 for connected graphs and is monotonically increasing. We can now reformulate Problem 1 as an optimization problem for $r_G$:

\begin{problem}[$k$-GRIP for the normalized effective graph resistance]
	Given an undirected, connected, simple graph $G = (V,E)$ and a non-negative integer $k$, find a subset $B \subseteq E^c$ of size $|B| = k$ which maximizes the normalized effective graph resistance $r_G(H)$ for the graph $H = (V, E \cup B)$.
	\label{problem_rG}
\end{problem}

It was shown by~\cite{summers2017correction} that the function $r_G$ is not submodular. In order to deal with non-submodular functions, the notion of submodularity can be extended as follows:


\begin{definition}[\cite{daskempe2011submodularityratio}]\label{def_submodularity_ratio}
	The \textbf{submodularity ratio} $\gamma$ of a function $f$ on a set $W$ is the largest $\gamma$ in the interval $[0,1]$ that satisfies
	\begin{equation*}
		f(S \cup \{v \}) - f(S) \geq \gamma (f(R \cup \{v\}) - f(R))
	\end{equation*}
	for all $S \subseteq R \subset V$ and all $v \in W \backslash R$.
\end{definition}
If $\gamma =1$, the function $f$ is submodular. If $0 < \gamma < 1$, the function~$f$ is called generalised submodular. The submodularity ratio $\gamma$ quantifies how close the metric $f$ is to being submodular. 

We will now show that for a given graph $G$, assuming that $W$ represents the set of $L^c$ non-existing links, computing $\gamma$ by verifying all possible choices for the sets $S, R$ and $W$, is very time-consuming. First, we require that $v \in W \backslash R$, so there are $L^c$ possible choices for $v$. Then only $L^c-1$ elements remain for $S$ and $W$. Now let $m$ be the size of the set $S$. Note that $S$ can be empty. It is required that $S$ and $R$ are not equal to $W$, because $v \in W \backslash R$. So, the size $m$ of the set $S$ runs from $m=0$ to $m=L^c-1$. For a given $m$, there are exactly $\binom{L^c-1}{m}$ possibilities to choose $m$ elements out of a total of $L^c-1$ elements. There are no further constraints for the set $S$. After choosing $S$, the set $R$ can still contain $L^c - m - 1$ elements. Thus, in total there are $2^{L^c - m - 1}$ possibilities for the set $R$. Then we conclude that
\begin{equation*}
	\#(S,R,v)= L^c \sum_{m=0}^{L^c-1} \binom{L^c-1}{m} 2^{L^c - m - 1} = L^c \cdot 3^{L^c - 1}
\end{equation*}
which grows extremely fast with the number of non-existent links $L^c$. For each combination of $S, R$ and $v$, the effective graph resistance must be computed. Using Eq.\ \eqref{eq_graph_res_laplacian}, the effective graph resistance can be computed in $\O(N^3)$ operations. Thus, verifying Definition~\ref{def_submodularity_ratio} requires $\O(N^3 \cdot L^c \cdot 3^{L^c-1})$ operations, which is infeasible for any moderately-sized graph.

We now state the greedy algorithm of $k$-GRIP for the normalized effective graph resistance $r_G$ in Algorithm~1. Greedy algorithm 1 runs in $\O( k \cdot W \cdot |f|)$ operations, where $|f|$ is the computation time of the metric~$f$. Calculating the effective graph resistance using Eq.\ \eqref{eq_graph_res_laplacian} requires $\O(N^3)$ operations. Thus, the greedy algorithm requires $\O(k N^5)$ operations for sparse graphs. 

\begin{algorithm}
	\label{alg_greedy}
	\caption{Greedy algorithm for $k$-GRIP.}
	\begin{algorithmic}[1]
		\State Given graph $G_1$, set of placeable links $W$ and number $k$ of links to be placed.
		\For{$i=1,\ldots, k$}
		\State $r_{G,\textnormal{opt}} \gets \infty$
		\State $W_{\textnormal{opt}} \gets \emptyset$
		\For{$j=1,\ldots, |W|$}
		\State Compute normalized effective graph resistance $r_G$ of $G_i \cup \{W_j\}$
		\If{$r_G(G_i \cup W_j) > r_{G,\textnormal{opt}}$}
		\State $W_{\textnormal{opt}} \gets W_j$
		\State $r_{G,\textnormal{opt}} \gets r_G(G_i \cup \{W_j\})$
		\EndIf
		\EndFor
		\State $G_{i+1} \gets G_i \cup \{V_{\textnormal{opt}}\}$
		\State $W \gets W \backslash W_j$
		\EndFor
		\State \textbf{Output:} $G_{k+1}$
	\end{algorithmic}
\end{algorithm}

Because $r_G$ is not submodular, the performance bound of~\cite{nemhauser1978subdmodularity}, i.e. the $(1-\frac{1}{e})$-closeness of the greedy solution to the optimal one, is not
guaranteed. However, the following result can be used for non-submodular functions:

\begin{theorem}[\cite{bian2017curvature}]\label{Bian}
	Let $f$ be a monotonically increasing, nonnegative function with submodularity ratio $\gamma\in[0,1]$ and curvature $\alpha\in[0,1]$. Then the greedy Algorithm 1 has the following guaranteed solution quality:
	\begin{equation}\label{eq_submodularity_quality}
		f_{\textnormal{greedy}} \geq \frac{1}{\alpha} \left( 1 - e^{-\gamma \alpha} \right) f_{\textnormal{optimal}}
	\end{equation}
\end{theorem}

Here, the curvature $\alpha$ is defined as follows:

\begin{definition}[\cite{bian2017curvature}]
	The \textbf{curvature} $\alpha$ of a function $f$ on a set $W$ is the smallest $\alpha$ in the interval $[0,1]$ that satisfies
	\begin{equation*}
		f(S \cup \Omega ) - f((S \backslash v) \cup \Omega) \geq (1-\alpha) (f(S) - f(S \backslash v))
	\end{equation*}
	for all $\Omega,S \subset W$ and all $v \in S \backslash \Omega$.
\end{definition}

The curvature $\alpha$ quantifies how close $f$ is to being supermodular. A function is called supermodular if $-f$ is submodular, where supermodularity corresponds to $\alpha =0$.

Recently \cite{liu2022} improved the result in Theorem \ref{Bian} as follows:

\begin{equation}\label{zhichengliu}
	f_{\textnormal{greedy}} \geq 
	(1 - (1 - \gamma + \gamma\alpha) e^{-\gamma}) f_{\textnormal{optimal}}
\end{equation}

Note that for $\alpha=1$ and $\gamma=1$, Eqs.\ \eqref{eq_submodularity_quality}--\eqref{zhichengliu} both lead to the performance guarantee provided by the submodularity condition. In case $\gamma = 0$, the Eqs.\ \eqref{eq_submodularity_quality}--\eqref{zhichengliu} both simplify to $f_{\textnormal{greedy}} \geq 0$, which does not provide any performance guarantee, because it was already assumed that $f$ was nonnegative. Note that for $f \equiv r_G$ and by using Eq.\ \eqref{eq_graph_resistance_scaled}, $f_{\textnormal{greedy}} \geq 0$ is equivalent to $R_G \leq \binom{N-1}{3}$, which indeed does not provide any new insight because it is already known that for connected graphs $\binom{N-1}{3}$ is an upper bound for $R_G$.

In this paper we will construct a family of graphs for which $\gamma \to 0$ for $N \to \infty$.


\section{A counterexample for generalized submodularity}
\label{sec_counterexample}
Summers \emph{et al}.\ \cite{summers2017correction} have already shown that the normalized effective graph resistance $r_G$ does not satisfy submodularity (see Definition \ref{def_submodularity}). Their counterexample, which is also the smallest possible counterexample, is a graph with $N=5$ nodes shown in Fig.\ \ref{fig_counterexample_easy}. Starting point is a graph $G$ with 5 nodes and 5 links, denoted as solid black lines. The set $S$ is the empty set, while the element $v$ is the link between nodes 1 and 2, represented by a dashed green line. The set $R$ is the link between nodes 2 and 3, represented by a dotted red line. 

To compute $\gamma$ for this case we need to determine the ratio of $r_G(G \cup \{v\}) - r_G(G)$ and 
$r_G(G \cup R \cup \{v\}) - r_G(G \cup R)$. Using Eq.\ \eqref{eq_graph_resistance_scaled}, this comes down to computing the ratio between $R_G(G) - R_G(G \cup \{v\})$ and $R_G(G \cup R) - R_G(G \cup R \cup \{v\})$. Evaluation of the 4 involved expressions gives $R_G(G) \approx 13.33, R_G(G \cup \{v\}) = 10.25, R_G(G \cup R) = 10.25$ and $R_G(G \cup R \cup \{v\}) \approx 6.95$.  
Thus the gain of adding the element $v$ (the dashed green link) to the original graph (the graph $G$) is approximately equal to $13.33 - 10.25 = 3.08$ while adding the element $v$ to the augmented graph (the graph $G$ with the link $R$ added), gives a larger gain of approximately $10.25 - 6.95 = 3.30$. Clearly this opposes the definition of submodularity. Thus, $k$-GRIP for the normalized effective graph resistance is not submodular. We can also deduce from the example that $\gamma \lesssim \frac{3.08}{3.30} = 0.935$.

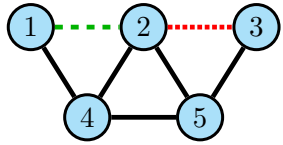
\begin{figure}[!ht]
	\centering
	\begin{tikzpicture}
		\node[mystyle, label=center:1] (1) at (-1.5,0) {};
		\node[mystyle, label=center:2] (2) at (0,0) {};
		\node[mystyle, label=center:3] (3) at (1.5,0) {};
		\node[mystyle, label=center:4] (4) at (-0.75,-1.2) {};
		\node[mystyle, label=center:5] (5) at (0.75,-1.2) {};
		
		\path [-,line width=1.8pt] (1) edge node[right] {} (4);
		\path [-,line width=1.8pt] (3) edge node[right] {} (5);
		\path [-,line width=1.8pt] (2) edge node[left] {} (4);
		\path [-,line width=1.8pt] (2) edge node[left] {} (5);
		\path [-,line width=1.8pt] (4) edge node[left] {} (5);
		\path [dashed,line width=1.8pt,color=black!30!green] (1) edge node[left] {} (2);
		\path [densely dotted,line width=1.8pt,color=red] (2) edge node[right] {} (3);
	\end{tikzpicture}
	\caption{The smallest counterexample showing that the normalized effective graph resistance $r_G$ in $k$-GRIP is not submodular. The graph $G$ consists of 5 nodes with 5 links, denoted as solid black lines. The set $S$ is the empty set, while the element $v$ is the link between nodes 1 and 2, represented by a dashed green line. The set $R$ is the link between nodes 2 and 3, represented by a dotted red line.}
	\label{fig_counterexample_easy}
\end{figure}

This result can be improved by showing that the normalized effective graph resistance even does not satisfy generalised submodularity. It is sufficient to construct a counterexample where $\gamma\to 0$. As a consequence, Eqs. \eqref{eq_submodularity_quality}--\eqref{zhichengliu} then imply that the greedy solution does not have any guaranteed performance.

To construct a counterexample we consider a graph $G$ on $2N$ nodes, $N \geq 4, N$ even. We start building the graph $G$ by first considering a complete bipartite graph $K_{2,N-2}$ on $N$ nodes. We define node $i$ and $j$ to be the nodes in the group with 2 nodes. Then, we select two nodes from the group with $N-2$ nodes and attach a path graph of length $N/2$ to each of those nodes. The resulting graph $G$ is depicted in Fig.\ \ref{fig_counterexample_S}. Again, we assume the set $S$ is the empty set.

The element $v$ is taken to be the link between node $i$ and $j$, as visualised in Fig.\ \ref{fig_counterexample_S_v}, represented by a dashed green line. The set $R$ contains two links: we connect the end of one path graph to node $i$ and the end of the other path graph to node $j$. The links in the set $R$ are represented by dotted red lines, see Fig.\ \ref{fig_counterexample_R}.
Finally, Fig.\ \ref{fig_counterexample_R_v} shows graph $G$ augmented with the sets $R$ and $v$.
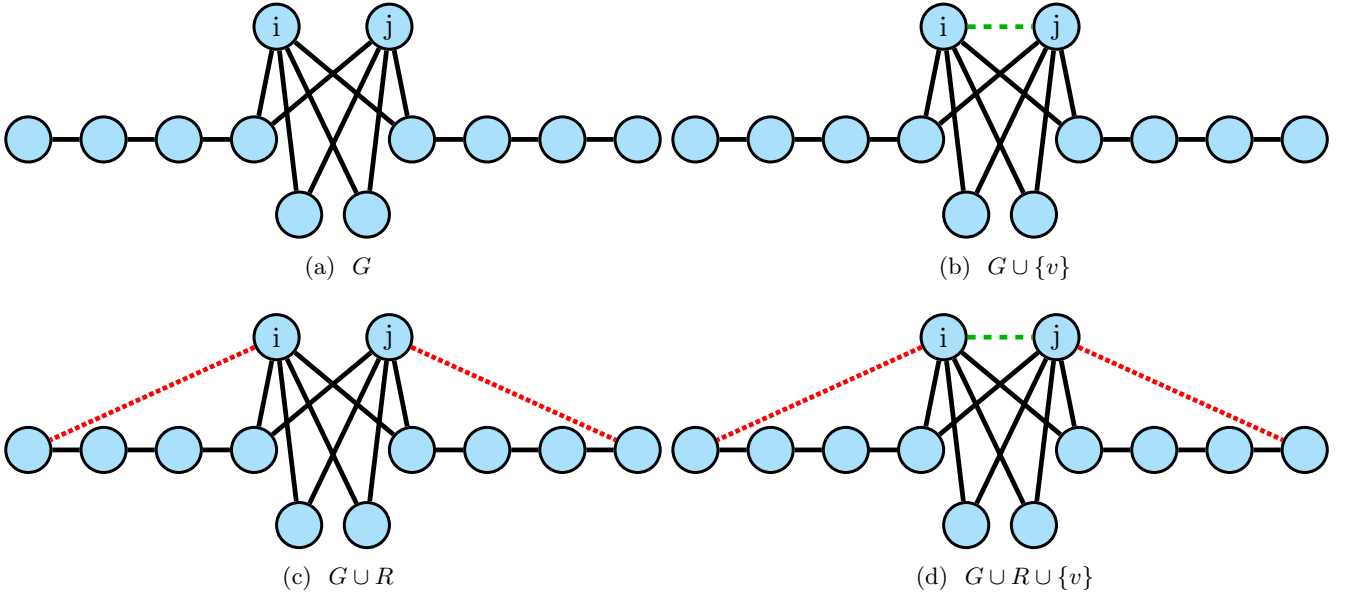
\begin{figure}[!ht]
	\centering
	\subfloat[\label{fig_counterexample_S} $G$]{%
		\begin{tikzpicture}
			\node[mystyle, label=center:i] (0) at (0,1.5) {};
			\node[mystyle, label=center:j] (1) at (1.5,1.5) {};
			\node[mystyle, label=center:] (2) at (1.8,0) {};
			\node[mystyle, label=center:] (3) at (-0.3,0) {};
			\node[mystyle, label=center:] (5) at (-1.3,0) {};
			\node[mystyle, label=center:] (8) at (-2.3,0) {};
			\node[mystyle, label=center:] (9) at (-3.3,0) {};
			\node[mystyle, label=center:] (6) at (2.8,0) {};
			\node[mystyle, label=center:] (7) at (3.8,0) {};
			\node[mystyle, label=center:] (10) at (4.8,0) {};
			\node[mystyle, label=center:] (11) at (0.3,-1) {};
			\node[mystyle, label=center:] (12) at (1.2,-1) {};
			
			\path [-,line width=1.8pt] (0) edge node[left] {} (2);
			\path [-,line width=1.8pt] (0) edge node[left] {} (3);
			\path [-,line width=1.8pt] (1) edge node[left] {} (2);
			\path [-,line width=1.8pt] (1) edge node[right] {} (3);
			\path [-,line width=1.8pt] (3) edge node[left] {} (5);
			\path [-,line width=1.8pt] (2) edge node[left] {} (6);
			\path [-,line width=1.8pt] (6) edge node[left] {} (7);
			\path [-,line width=1.8pt] (5) edge node[left] {} (8);
			\path [-,line width=1.8pt] (8) edge node[left] {} (9);
			\path [-,line width=1.8pt] (7) edge node[left] {} (10);
			\path [-,line width=1.8pt] (0) edge node[left] {} (11);
			\path [-,line width=1.8pt] (0) edge node[left] {} (12);
			\path [-,line width=1.8pt] (1) edge node[left] {} (11);
			\path [-,line width=1.8pt] (1) edge node[left] {} (12);
		\end{tikzpicture}
	}
	\subfloat[\label{fig_counterexample_S_v} $G \cup \{v\}$]{%
		\begin{tikzpicture}
			\node[mystyle, label=center:i] (0) at (0,1.5) {};
			\node[mystyle, label=center:j] (1) at (1.5,1.5) {};
			\node[mystyle, label=center:] (2) at (1.8,0) {};
			\node[mystyle, label=center:] (3) at (-0.3,0) {};
			\node[mystyle, label=center:] (5) at (-1.3,0) {};
			\node[mystyle, label=center:] (8) at (-2.3,0) {};
			\node[mystyle, label=center:] (9) at (-3.3,0) {};
			\node[mystyle, label=center:] (6) at (2.8,0) {};
			\node[mystyle, label=center:] (7) at (3.8,0) {};
			\node[mystyle, label=center:] (10) at (4.8,0) {};
			\node[mystyle, label=center:] (11) at (0.3,-1) {};
			\node[mystyle, label=center:] (12) at (1.2,-1) {};
			
			\path [dashed,line width=1.8pt,color=black!30!green] (0) edge node[left] {} (1);
			\path [-,line width=1.8pt] (0) edge node[left] {} (2);
			\path [-,line width=1.8pt] (0) edge node[left] {} (3);
			\path [-,line width=1.8pt] (1) edge node[left] {} (2);
			\path [-,line width=1.8pt] (1) edge node[right] {} (3);
			\path [-,line width=1.8pt] (3) edge node[left] {} (5);
			\path [-,line width=1.8pt] (2) edge node[left] {} (6);
			\path [-,line width=1.8pt] (6) edge node[left] {} (7);
			\path [-,line width=1.8pt] (5) edge node[left] {} (8);
			\path [-,line width=1.8pt] (8) edge node[left] {} (9);
			\path [-,line width=1.8pt] (7) edge node[left] {} (10);
			\path [-,line width=1.8pt] (0) edge node[left] {} (11);
			\path [-,line width=1.8pt] (0) edge node[left] {} (12);
			\path [-,line width=1.8pt] (1) edge node[left] {} (11);
			\path [-,line width=1.8pt] (1) edge node[left] {} (12);
		\end{tikzpicture}
	} \\
	\subfloat[\label{fig_counterexample_R} $G \cup R$]{%
		\begin{tikzpicture}
			\node[mystyle, label=center:i] (0) at (0,1.5) {};
			\node[mystyle, label=center:j] (1) at (1.5,1.5) {};
			\node[mystyle, label=center:] (2) at (1.8,0) {};
			\node[mystyle, label=center:] (3) at (-0.3,0) {};
			\node[mystyle, label=center:] (5) at (-1.3,0) {};
			\node[mystyle, label=center:] (8) at (-2.3,0) {};
			\node[mystyle, label=center:] (9) at (-3.3,0) {};
			\node[mystyle, label=center:] (6) at (2.8,0) {};
			\node[mystyle, label=center:] (7) at (3.8,0) {};
			\node[mystyle, label=center:] (10) at (4.8,0) {};
			\node[mystyle, label=center:] (11) at (0.3,-1) {};
			\node[mystyle, label=center:] (12) at (1.2,-1) {};
			
			\path [-,line width=1.8pt] (0) edge node[left] {} (2);
			\path [-,line width=1.8pt] (0) edge node[left] {} (3);
			\path [-,line width=1.8pt] (1) edge node[left] {} (2);
			\path [-,line width=1.8pt] (1) edge node[right] {} (3);
			\path [-,line width=1.8pt] (3) edge node[left] {} (5);
			\path [-,line width=1.8pt] (2) edge node[left] {} (6);
			\path [-,line width=1.8pt] (6) edge node[left] {} (7);
			\path [-,line width=1.8pt] (5) edge node[left] {} (8);
			\path [-,line width=1.8pt] (8) edge node[left] {} (9);
			\path [-,line width=1.8pt] (7) edge node[left] {} (10);
			\path [-,line width=1.8pt] (0) edge node[left] {} (11);
			\path [-,line width=1.8pt] (0) edge node[left] {} (12);
			\path [-,line width=1.8pt] (1) edge node[left] {} (11);
			\path [-,line width=1.8pt] (1) edge node[left] {} (12);
			
			\path [densely dotted,color=red,line width=1.8pt] (0) edge node[left] {} (9);
			\path [densely dotted,color=red,line width=1.8pt] (1) edge node[left] {} (10);
		\end{tikzpicture}
	}
	\subfloat[\label{fig_counterexample_R_v} $G \cup R \cup \{v\}$]{%
		\begin{tikzpicture}
			\node[mystyle, label=center:i] (0) at (0,1.5) {};
			\node[mystyle, label=center:j] (1) at (1.5,1.5) {};
			\node[mystyle, label=center:] (2) at (1.8,0) {};
			\node[mystyle, label=center:] (3) at (-0.3,0) {};
			\node[mystyle, label=center:] (5) at (-1.3,0) {};
			\node[mystyle, label=center:] (8) at (-2.3,0) {};
			\node[mystyle, label=center:] (9) at (-3.3,0) {};
			\node[mystyle, label=center:] (6) at (2.8,0) {};
			\node[mystyle, label=center:] (7) at (3.8,0) {};
			\node[mystyle, label=center:] (10) at (4.8,0) {};
			\node[mystyle, label=center:] (11) at (0.3,-1) {};
			\node[mystyle, label=center:] (12) at (1.2,-1) {};
			
			\path [dashed,line width=1.8pt,color=black!30!green] (0) edge node[left] {} (1);
			\path [-,line width=1.8pt] (0) edge node[left] {} (2);
			\path [-,line width=1.8pt] (0) edge node[left] {} (3);
			\path [-,line width=1.8pt] (1) edge node[left] {} (2);
			\path [-,line width=1.8pt] (1) edge node[right] {} (3);
			\path [-,line width=1.8pt] (3) edge node[left] {} (5);
			\path [-,line width=1.8pt] (2) edge node[left] {} (6);
			\path [-,line width=1.8pt] (6) edge node[left] {} (7);
			\path [-,line width=1.8pt] (5) edge node[left] {} (8);
			\path [-,line width=1.8pt] (8) edge node[left] {} (9);
			\path [-,line width=1.8pt] (7) edge node[left] {} (10);
			\path [-,line width=1.8pt] (0) edge node[left] {} (11);
			\path [-,line width=1.8pt] (0) edge node[left] {} (12);
			\path [-,line width=1.8pt] (1) edge node[left] {} (11);
			\path [-,line width=1.8pt] (1) edge node[left] {} (12);
			
			\path [densely dotted,color=red,line width=1.8pt] (0) edge node[left] {} (9);
			\path [densely dotted,color=red,line width=1.8pt] (1) edge node[left] {} (10);
		\end{tikzpicture}
	}
	\caption{The graph $G$ consisting of $2N$ nodes; two path graphs with $N/2$ nodes are attached to a complete bipartite graph $K_{2,N-2}$ on $N$ nodes. The element $v$ is the dashed green link in the complete bipartite graph between node $i$ and $j$. The set $R$, represented by dotted red links, is the union of the link connecting node $i$ with the left-most node of the graph and the link connecting node $j$ with the right-most node of the graph. In this example, $N=6$.}
	\label{fig_counterexample}
\end{figure}

We now present our main results. 

\begin{theorem}\label{thm_counterexample}
	Consider the graph $G$ depicted in Fig.\ \ref{fig_counterexample}, the element $v$ and the set $R$. Then the following holds 
	\begin{equation}\label{S}
		R_G(G) - R_G(G \cup \{ v \})  = \frac{4}{N-2}
	\end{equation}
	and 
	\begin{equation}\label{R}
		R_G(G \cup R) - R_G(G \cup R \cup \{ v\})  = \frac{2 \,(N+3)\,(N+4)\,(N+5)}{3\,(N+1)\,(N+2)\,(N^2+N-4)}
	\end{equation}
\end{theorem}
\begin{proof}
	See Appendix~\ref{app_proof}.
\end{proof}

\begin{theorem}\label{thm_counterexample2}
	Consider the graph $G$ depicted in Fig.\ \ref{fig_counterexample}. Then the submodularity ratio for $k$-GRIP for the normalized effective graph resistance $r_G$ satisfies
	\begin{equation}\label{eq_counterexample_gamma}
		\gamma \leq \frac{6\,(N+1)\,(N+2)\,(N^2+N-4)}{(N-2)\, N\, (N+3)\,(N+4)\,(N+5)}
	\end{equation}
	which, for large $N$ converges to zero. In other words, the normalized effective graph resistance $r_G$ is not generalised submodular and the accuracy bounds from Eqs.\ \eqref{eq_submodularity_quality}--\eqref{zhichengliu} do not provide any guaranteed performance for the greedy algorithm.
\end{theorem}

\begin{proof}
	We already observed that in order to estimate $\gamma$  for the normalized effective graph resistance we need to determine the ratio of $r_G(G \cup \{v\}) - r_G(G)$ and $r_G(G \cup R \cup \{v\}) - r_G(G \cup R)$, which equals the ratio between $R_G(G) - R_G(G \cup \{v\})$ and $R_G(G \cup R) - R_G(G \cup R \cup \{v\})$, according to Eq.\ \eqref{eq_graph_resistance_scaled}. Hence, Eq.\ \eqref{eq_counterexample_gamma} follows from  Eqs.\ \eqref{S}--\eqref{R}.
\end{proof}
Note that the right-hand side of Eq.\ \eqref{eq_counterexample_gamma} asymptotically behaves as $6/N$. Fig.\ \ref{fig_path_graph_performance} depicts the exact upper bound for $\gamma$ and its asymptote. For $2N \geq 100$, both curves are nearly indistinguishable.

\begin{figure}[H]
	\centering
	\includegraphics[width=0.7\columnwidth]{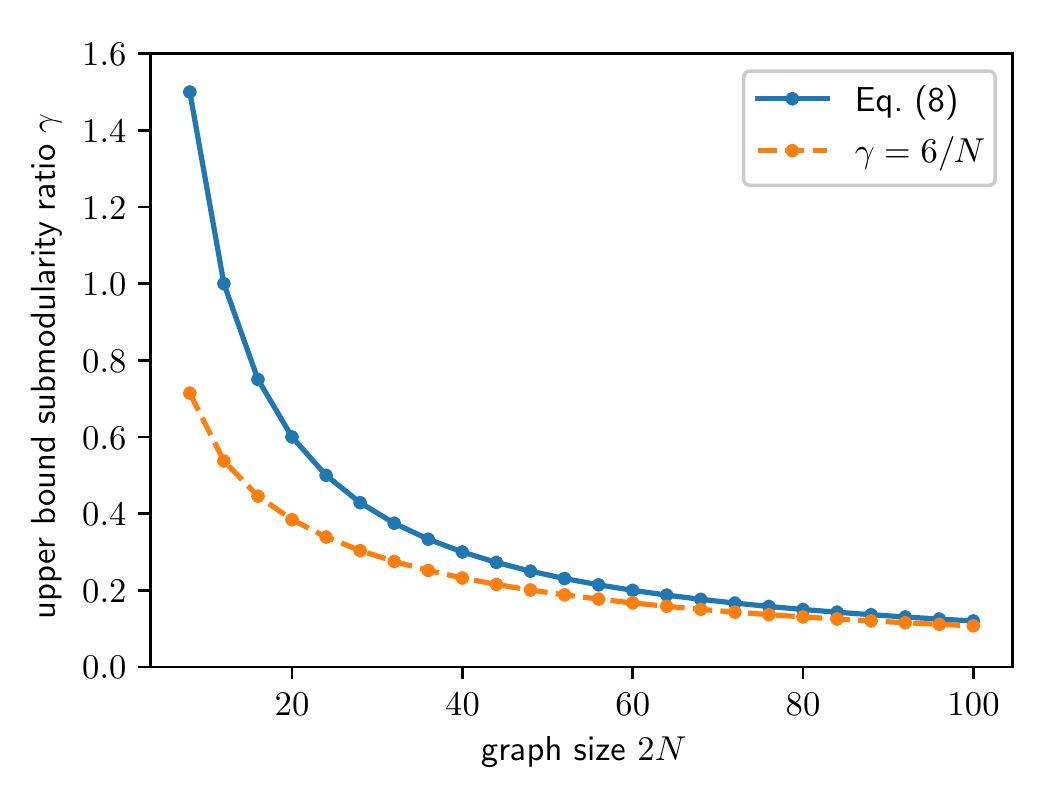}
	\caption{The upper bound for the submodularity ratio~$\gamma$ according to Eq.\ \eqref{eq_counterexample_gamma} and its asymptote $6/N$ for various values of $N$.}
	\label{fig_path_graph_performance}
\end{figure}

\section{The quality of the greedy algorithm}\label{sec_greedy}
The counterexample from Theorem~\ref{thm_counterexample2} demonstrates that the quality of the greedy solution cannot be guaranteed. However, \emph{this does not necessarily imply that the greedy algorithm actually performs bad}. Instead, we just cannot guarantee the quality of the greedy algorithm.

In this section, we measure the performance of the greedy algorithm. We define the efficiency $\eta$ as
\begin{equation*}
	\eta = \frac{(R_G)_{\textnormal{opt}}}{(R_G)_{\textnormal{greedy}}}
\end{equation*}
If the efficiency $\eta=1$, the greedy algorithm provides the optimal solution. If the efficiency $\eta < 1$, the greedy algorithm produces a sub-optimal solution. 

We determine the accuracy of the greedy algorithm by looking at small graphs. We consider $2 \leq k \leq 6$ links to be added by the greedy algorithm, which we then compare with the optimal value, which we obtain by brute-force. For $k=1$, the optimal and the greedy algorithm coincide. We generate all non-isomorphic, connected graphs using Nauty and Traces \cite{nautytraces}. Out of all non-isomorphic connected graphs, we compute the smallest efficiency $\eta_{\min}$ for $5 \leq N \leq 10$ and $2 \leq k \leq 6$, which are shown in Table~\ref{table_greedy2}. The computational complexity equals $\O(k L^c N^3)$ for the greedy algorithm and $\O((L^c)^k N^3)$ for the brute-force algorithm, implying that the procedure can only be executed for small graphs. We verified all graphs with $N \leq 8$ nodes and randomly selected several graphs for $N=9, N=10$. The worst case, corresponding to the graph shown in Fig.\ \ref{fig_worst_example} with $N=10$ nodes and $k=3$ added links, has efficiency $\eta=\mathbf{0.878}$, implying that the greedy algorithm produces a graph whose effective graph resistance is a factor 0.878 different from the optimal effective graph resistance. Already for a small graph with $N=10$ nodes, the efficiency $\eta$ is quite low. We expect that larger graphs may have even lower efficiencies, but the computational complexity is too demanding to demonstrate that here.


\begin{table}[H]
	\centering
	\caption{The lowest efficiency $\eta_{\min}$ for each possible combination of the number of nodes $N$ and the number of added links $k$. The values for $N=9,10$ are lower bounds, as the number of non-isomorphic graphs is too large to compute all; instead, we randomly selected several graphs. The lowest efficiency $\eta_{\min}$ is highlighted in \textbf{bold}.}
	\begin{tabular}{c|c|c|c|c|c|c}
		\multicolumn{2}{c}{} \\
		& $N=5$ & $N=6$ & $N=7$ & $N=8$ & $N=9$ & $N=10$ \\
		$k$ & 21 graphs & 112 graphs & 853 graphs & 11,117 graphs & 261,080 graphs & 11,716,571 graphs \\
		\hline
		2 & 0.937 & 0.946 & 0.934 & 0.922 & $0.891$ & $\leq 0.905$ \\
		3 & 1 & 0.957 & 0.940 & 0.931 & $0.913$ & $\mathbf{\leq 0.878}$ \\
		4 & 1 & 0.957 & 0.949 & 0.936 & $0.899$ & $\leq 0.918$ \\
		5 & 1 & 0.965 & 0.966 & 0.949 & $\leq 0.919$ & $\leq 0.918$\\
		6 & 1 & 0.970 & 0.966 & 0.950 &  $\leq 0.941$ & $\leq 0.932$
	\end{tabular}
	\label{table_greedy2}
\end{table}

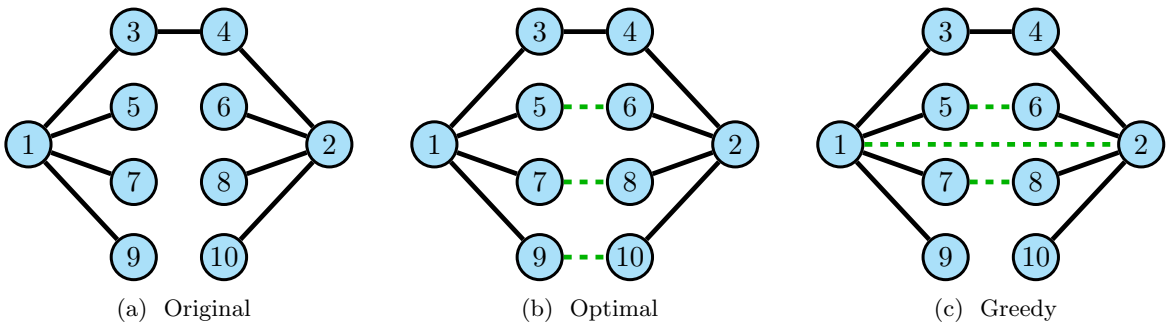
\begin{figure}[!ht]
	\centering
	\subfloat[\label{fig_worst1} Original]{%
		\begin{tikzpicture}
			\node[mystyle, label=center:1] (N1) at (-2,-1.5) {};
			\node[mystyle, label=center:2] (N2) at (2,-1.5) {};
			\node[mystyle, label=center:3] (N3) at (-0.6,0) {};
			\node[mystyle, label=center:4] (N4) at (0.6,0) {};
			\node[mystyle, label=center:5] (N5) at (-0.6, -1) {};
			\node[mystyle, label=center:6] (N6) at (0.6, -1) {};
			\node[mystyle, label=center:7] (N7) at (-0.6, -2) {};
			\node[mystyle, label=center:8] (N8) at (0.6, -2) {};
			\node[mystyle, label=center:9] (N9) at (-0.6,-3) {};
			\node[mystyle, label=center:10] (N10) at (0.6,-3) {};
			
			\path [-,line width=1.8pt] (N1) edge node[right] {} (N3);
			\path [-,line width=1.8pt] (N1) edge node[left] {} (N5);
			\path [-,line width=1.8pt] (N1) edge node[left] {} (N7);
			\path [-,line width=1.8pt] (N1) edge node[left] {} (N9);
			\path [-,line width=1.8pt] (N2) edge node[left] {} (N4);
			\path [-,line width=1.8pt] (N2) edge node[left] {} (N6);
			\path [-,line width=1.8pt] (N2) edge node[left] {} (N8);
			\path [-,line width=1.8pt] (N2) edge node[left] {} (N10);
			\path [-,line width=1.8pt] (N3) edge node[right] {} (N4);
			
			
			
		\end{tikzpicture}
	}\hspace{5mm} \subfloat[\label{fig_worst2} Optimal]{%
		\begin{tikzpicture}
			\node[mystyle, label=center:1] (N1) at (-2,-1.5) {};
			\node[mystyle, label=center:2] (N2) at (2,-1.5) {};
			\node[mystyle, label=center:3] (N3) at (-0.6,0) {};
			\node[mystyle, label=center:4] (N4) at (0.6,0) {};
			\node[mystyle, label=center:5] (N5) at (-0.6, -1) {};
			\node[mystyle, label=center:6] (N6) at (0.6, -1) {};
			\node[mystyle, label=center:7] (N7) at (-0.6, -2) {};
			\node[mystyle, label=center:8] (N8) at (0.6, -2) {};
			\node[mystyle, label=center:9] (N9) at (-0.6,-3) {};
			\node[mystyle, label=center:10] (N10) at (0.6,-3) {};
			
			\path [-,line width=1.8pt] (N1) edge node[right] {} (N3);
			\path [-,line width=1.8pt] (N1) edge node[left] {} (N5);
			\path [-,line width=1.8pt] (N1) edge node[left] {} (N7);
			\path [-,line width=1.8pt] (N1) edge node[left] {} (N9);
			\path [-,line width=1.8pt] (N2) edge node[left] {} (N4);
			\path [-,line width=1.8pt] (N2) edge node[left] {} (N6);
			\path [-,line width=1.8pt] (N2) edge node[left] {} (N8);
			\path [-,line width=1.8pt] (N2) edge node[left] {} (N10);
			\path [-,line width=1.8pt] (N3) edge node[right] {} (N4);
			
			\path [dashed,line width=1.8pt,color=black!30!green] (N5) edge node[left] {} (N6);
			\path [dashed,line width=1.8pt,color=black!30!green] (N7) edge node[left] {} (N8);
			\path [dashed,line width=1.8pt,color=black!30!green] (N9) edge node[left] {} (N10);
			
			
		\end{tikzpicture}
	}\hspace{5mm} \subfloat[\label{fig_worst3} Greedy]{%
		\begin{tikzpicture}
			\node[mystyle, label=center:1] (N1) at (-2,-1.5) {};
			\node[mystyle, label=center:2] (N2) at (2,-1.5) {};
			\node[mystyle, label=center:3] (N3) at (-0.6,0) {};
			\node[mystyle, label=center:4] (N4) at (0.6,0) {};
			\node[mystyle, label=center:5] (N5) at (-0.6, -1) {};
			\node[mystyle, label=center:6] (N6) at (0.6, -1) {};
			\node[mystyle, label=center:7] (N7) at (-0.6, -2) {};
			\node[mystyle, label=center:8] (N8) at (0.6, -2) {};
			\node[mystyle, label=center:9] (N9) at (-0.6,-3) {};
			\node[mystyle, label=center:10] (N10) at (0.6,-3) {};
			
			\path [-,line width=1.8pt] (N1) edge node[right] {} (N3);
			\path [-,line width=1.8pt] (N1) edge node[left] {} (N5);
			\path [-,line width=1.8pt] (N1) edge node[left] {} (N7);
			\path [-,line width=1.8pt] (N1) edge node[left] {} (N9);
			\path [-,line width=1.8pt] (N2) edge node[left] {} (N4);
			\path [-,line width=1.8pt] (N2) edge node[left] {} (N6);
			\path [-,line width=1.8pt] (N2) edge node[left] {} (N8);
			\path [-,line width=1.8pt] (N2) edge node[left] {} (N10);
			\path [-,line width=1.8pt] (N3) edge node[right] {} (N4);
			
			
			\path [dashed,line width=1.8pt,color=black!30!green] (N5) edge node[left] {} (N6);
			\path [dashed,line width=1.8pt,color=black!30!green] (N1) edge node[left] {} (N2);
			\path [dashed,line width=1.8pt,color=black!30!green] (N7) edge node[left] {} (N8);
			
		\end{tikzpicture}
	}
	\caption{The graph $G$ with the currently known smallest efficiency $\eta_{\min}=0.878$ on $N=10$ nodes and the $k=3$ added links are shown as dashed green links.}
	\label{fig_worst_example}
\end{figure}

\section{Conclusion}
\label{sec_conclusion}
We consider the optimisation problem of adding $k$ links to a given network, such that the resulting normalized effective graph resistance is as large as possible. We have shown that this problem is not generalised submodular, by providing an example for which the submodularity ratio converges to zero. As a consequence, we have no guaranteed solution quality of the greedy algorithm. We additionally investigate the efficiency of the greedy algorithm for optimizing the effective graph resistance for small graphs. Already for graphs with 10 nodes, the efficiency can be as low as $0.878$. On the agenda for future research is finding an example where the performance of the greedy algorithm is even lower. One possible method is to consider a class of graphs which, for a given number of nodes $N$ and number of links $L$, have optimal effective graph resistance. Then we start with another graph with $N$ nodes and a subset of the links in the first graph. By sequentially adding links using the greedy algorithm, the graph may converge to the optimum, or if otherwise, we may find explicit examples for the efficiency smaller than one in large graphs. Alternatively, one can derive alternative lower bounds on the solution quality of the greedy algorithm.

%

\clearpage

\bibliographystyle{unsrt}

\appendix

\section{Proof of Theorem \ref{thm_counterexample}}
\label{app_proof}
Prior to our proof, we first present two powerful theorems.
\begin{theorem}[Theorem 2.1 from \cite{yang2013graphresistance}]\label{thm_link_weight}
	Let $\omega$ and $\omega'$ be resistance distance functions for connected graphs $G$ and $G'$ which are the same except for the weights $w$ and $w'$ on an link with nodes $i$ and $j$. Introduce $\delta = w'-w$. Then for any nodes $p \neq q$ it holds that
	\begin{equation}\label{eq_omega_diff}
		\omega'_{pq} = \omega_{pq} - \frac{\delta \cdot \left[ \omega_{p,i} + \omega_{q,j} - \omega_{p,j} - \omega_{q,i} \right]^2}{4[1+\delta \omega_{ij}]}
	\end{equation}
\end{theorem}

\begin{theorem}[Theorem 4.1 from \cite{yang2013graphresistance}]
	Let $\omega$ and $\omega'$ be resistance distance functions for connected, weighted graphs $G$ and $G'$ on $\tilde{N}$ nodes which are the same, except for the weights $w$ and $w'$ on an link with node $i$ and $j$. Introduce $\delta=w'-w$. Then
	\begin{equation}\label{eq_RG_diff}
		R_G(G)-R_G(G') = \frac{\delta \tilde{N}\sum_{k=1}^{\tilde{N}} ( \omega_{ik} - \omega_{jk})^2 - \delta \left[ \sum_{k=1}^{\tilde{N}} \omega_{ik} - \sum_{k=1}^{\tilde{N}} \omega_{jk} \right]^2}{4(1+\delta \omega_{ij})}
	\end{equation}
	\label{thm_add_link}
\end{theorem}

\subsection{The graphs $G$ and $G \cup \{v\}$}
The graph $G$ is a graph on $2N$ nodes. The graph can be constructed from a complete bipartite graph $K_{2,N-2}$ as follows. Select two nodes from the group with $N-2$ nodes and attach to each of these nodes a path graph of length $N/2$. Now denote the two nodes that belong to the group consisting of two nodes only, as node $i$ and $j$. The element $v$ is the link between node $i$ and $j$. For a visualisation, see Fig.\ \ref{fig_set_S}.

\begin{figure}[H]
	\centering
	\begin{tikzpicture}
		\node[mystyle, label=center:i] (0) at (0,1.5) {};
		\node[mystyle, label=center:j] (1) at (1.5,1.5) {};
		\node[mystyle, label=center:] (2) at (1.8,0) {};
		\node[mystyle, label=center:] (3) at (-0.3,0) {};
		\node[mystyle, label=center:] (5) at (-1.5,0) {};
		\node[mystyle, label=center:] (8) at (-3,0) {};
		\node[mystyle, label=center:] (9) at (-4.5,0) {};
		\node[mystyle, label=center:] (6) at (3,0) {};
		\node[mystyle, label=center:] (7) at (4.5,0) {};
		\node[mystyle, label=center:] (10) at (6,0) {};
		\node[mystyle, label=center:] (11) at (0.3,-1) {};
		\node[mystyle, label=center:] (12) at (1.2,-1) {};
		
		\path [dashed,line width=1.8pt,color=black!30!green] (0) edge node[left] {} (1);
		\path [-,line width=1.8pt] (0) edge node[left] {} (2);
		\path [-,line width=1.8pt] (0) edge node[left] {} (3);
		\path [-,line width=1.8pt] (1) edge node[left] {} (2);
		\path [-,line width=1.8pt] (1) edge node[right] {} (3);
		\path [-,line width=1.8pt] (3) edge node[left] {} (5);
		\path [-,line width=1.8pt] (2) edge node[left] {} (6);
		\path [-,line width=1.8pt] (6) edge node[left] {} (7);
		\path [-,line width=1.8pt] (5) edge node[left] {} (8);
		\path [-,line width=1.8pt] (8) edge node[left] {} (9);
		\path [-,line width=1.8pt] (7) edge node[left] {} (10);
		\path [-,line width=1.8pt] (0) edge node[left] {} (11);
		\path [-,line width=1.8pt] (0) edge node[left] {} (12);
		\path [-,line width=1.8pt] (1) edge node[left] {} (11);
		\path [-,line width=1.8pt] (1) edge node[left] {} (12);
	\end{tikzpicture}
	\caption{The original graph $G$. The element $v$ is the dashed green link between node $i$ and $j$.}
	\label{fig_set_S}
\end{figure}
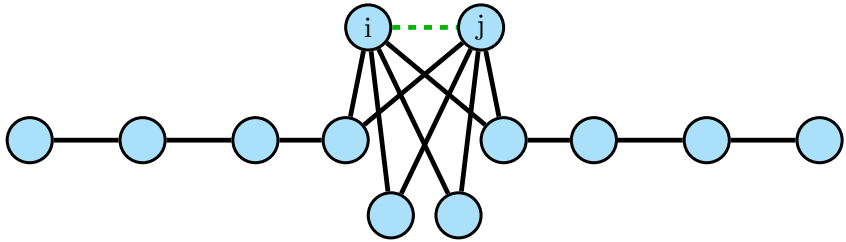
We directly apply Theorem~\ref{thm_add_link} with $\delta=1$ and $\tilde{N}=2N$, where the graph $G$ is considered the \quotes{original graph} and graph $G'$ is graph $G$ augmented with $v$. Then we find
\begin{equation}\label{eq_hfksyfisyf}
	R_G(G) - R_G(G \cup \{ v \})   = \frac{ 2N \sum_{k=1}^{2N} ( \omega_{ik} - \omega_{jk})^2 - \left[ \sum_{k=1}^{2N} \omega_{ik} - \sum_{k=1}^{2N} \omega_{jk} \right]^2}{4(1+ \omega_{ij})}
\end{equation}
Due to symmetry, we can immediately conclude that the second term is zero. Moreover, $\omega_{ik} = \omega_{jk}$ for all $k \leq i$ and $k \neq j$, again due to symmetry. Eq.\ (\ref{eq_hfksyfisyf}) then simplifies to
\begin{equation}
	R_G(G) - R_G(G \cup \{ v \}) = \frac{ N \omega_{ij}^2}{1+ \omega_{ij}}
\end{equation}
Since $\omega_{ij}$ is not affected by the path graphs, the question of finding $\omega_{ij}$ simplifies to finding the effective resistance between two nodes in the complete bipartite graph $K_{2,N-2}$. Using the theory of parallel and series resistors, we find that
\begin{equation*}
	\omega_{ij} = \frac{2}{N-2}
\end{equation*}
and conclude that
\begin{equation}\label{eq_deltaS}
	R_G(G) - R_G(G \cup \{ v \}) = \frac{4}{N-2}
\end{equation}

\subsection{The graphs $G \cup R$ and $G \cup R \cup \{v\}$}
The set $R$ consists of two links. From node $i$, we add one link to the end of a path graph, i,e, the left-most node of graph $G$ and perform the same procedure to node $j$. The resulting graph $G \cup R$ is shown in Fig.\ \ref{fig_set_R}.

\begin{figure}[H]
	\centering
	\begin{tikzpicture}
		\node[mystyle, label=center:i] (0) at (0,1.5) {};
		\node[mystyle, label=center:j] (1) at (1.5,1.5) {};
		\node[mystyle, label=center:m] (2) at (1.8,0) {};
		\node[mystyle, label=center:l] (3) at (-0.3,0) {};
		\node[mystyle, label=center:{\scriptsize N/2}] (5) at (-1.5,0) {};
		\node[mystyle, label=center:k] (8) at (-3,0) {};
		\node[mystyle, label=center:1] (9) at (-4.5,0) {};
		\node[mystyle, label=center:] (6) at (3,0) {};
		\node[mystyle, label=center:] (7) at (4.5,0) {};
		\node[mystyle, label=center:] (10) at (6,0) {};
		\node[mystyle, label=center:] (11) at (0.3,-1) {};
		\node[mystyle, label=center:] (12) at (1.2,-1) {};
		
		\path [-,line width=1.8pt] (0) edge node[left] {} (2);
		\path [-,line width=1.8pt] (0) edge node[left] {} (3);
		\path [-,line width=1.8pt] (1) edge node[left] {} (2);
		\path [-,line width=1.8pt] (1) edge node[right] {} (3);
		\path [-,line width=1.8pt] (3) edge node[left] {} (5);
		\path [-,line width=1.8pt] (2) edge node[left] {} (6);
		\path [-,line width=1.8pt] (6) edge node[left] {} (7);
		\path [-,dotted,line width=1.8pt] (5) edge node[left] {} (8);
		\path [-,dotted,line width=1.8pt] (8) edge node[left] {} (9);
		\path [-,line width=1.8pt] (7) edge node[left] {} (10);
		\path [-,line width=1.8pt] (0) edge node[left] {} (11);
		\path [-,line width=1.8pt] (0) edge node[left] {} (12);
		\path [-,line width=1.8pt] (1) edge node[left] {} (11);
		\path [-,line width=1.8pt] (1) edge node[left] {} (12);
		
		\path [densely dotted,line width=1.8pt,color=red] (1) edge node[left] {} (10);
		\path [densely dotted,line width=1.8pt,color=red] (0) edge node[left] {} (9);
	\end{tikzpicture}
	\caption{The augmented graph $G \cup R$, which contains the original graph $G$ and the set $R$, denoted by dotted red links.}
	\label{fig_set_R}
\end{figure}
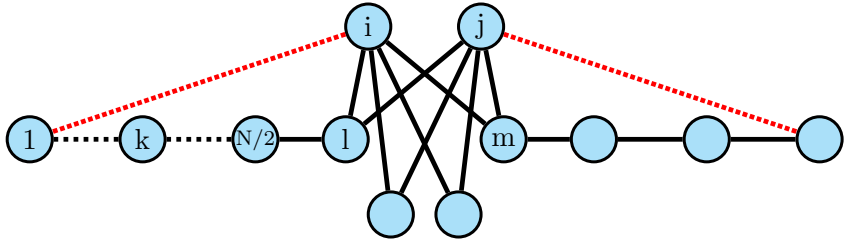

We now return to Theorem~\ref{thm_add_link} with $\delta=1$ and $\tilde{N}=2N$, where the graph $G \cup R$ is considered the \quotes{original graph}. Then we find
\begin{equation}
	R_G(G \cup R) - R_G(G \cup R \cup \{ v \})  = \frac{ 2N \sum_{k=1}^{2N} ( \omega_{ik} - \omega_{jk})^2 - \left[ \sum_{k=1}^{2N} \omega_{ik} - \sum_{k=1}^{2N} \omega_{jk} \right]^2}{4(1+ \omega_{ij})}
\end{equation}
The second term is again zero due symmetry. Then the equation simplifies to
\begin{equation}\label{eq_R_delta}
	\Delta R_G = \frac{2N}{4(1+\omega_{ij})} \sum_{k=1}^{2N} (\omega_{ik}-\omega_{jk})^2
\end{equation}
In order to determine $\Delta R_G$ we need to find expressions for the effective resistance $\omega_{ij}, \omega_{ik}$ and $\omega_{jk}$ for the graph $G \cup R$. Note that the graph $R$ is planar, as can be seen from Fig.\ \ref{Rplanar}, which shows the graph $G \cup R$ drawn without intersecting links.

\begin{figure}[H]
	\centering
	\begin{tikzpicture}
		\node[mystyle, label=center:i] (0) at (-0.3,0) {};
		\node[mystyle, label=center:j] (1) at (1.8,0) {};
		\node[mystyle, label=center:m] (2) at (0.75,-1.5) {};
		\node[mystyle, label=center:l] (3) at (0.75,1.5) {};
		\node[mystyle, label=center:{\scriptsize N/2}] (5) at (-4.5,0) {};
		\node[mystyle, label=center:k] (8) at (-3,0) {};
		\node[mystyle, label=center:1] (9) at (-1.5,0) {};
		\node[mystyle, label=center:] (6) at (6,0) {};
		\node[mystyle, label=center:] (7) at (4.5,0) {};
		\node[mystyle, label=center:] (10) at (3,0) {};
		\node[mystyle, label=center:] (11) at (0.75,0.5) {};
		\node[mystyle, label=center:] (12) at (0.75,-0.5) {};
		
		\path [-,line width=1.8pt] (0) edge node[left] {} (2);
		\path [-,line width=1.8pt] (0) edge node[left] {} (3);
		\path [-,line width=1.8pt] (1) edge node[left] {} (2);
		\path [-,line width=1.8pt] (1) edge node[right] {} (3);
		\path [-,line width=1.8pt] (3) edge node[left] {} (5);
		\path [-,line width=1.8pt] (2) edge node[left] {} (6);
		\path [-,line width=1.8pt] (6) edge node[left] {} (7);
		\path [-,dotted,line width=1.8pt] (5) edge node[left] {} (8);
		\path [-,dotted,line width=1.8pt] (8) edge node[left] {} (9);
		\path [-,line width=1.8pt] (7) edge node[left] {} (10);
		\path [-,line width=1.8pt] (0) edge node[left] {} (11);
		\path [-,line width=1.8pt] (0) edge node[left] {} (12);
		\path [-,line width=1.8pt] (1) edge node[left] {} (11);
		\path [-,line width=1.8pt] (1) edge node[left] {} (12);
		
		\path [densely dotted,line width=1.8pt,color=red] (1) edge node[left] {} (10);
		\path [densely dotted,line width=1.8pt,color=red] (0) edge node[left] {} (9);
	\end{tikzpicture}
	\caption{The graph $G \cup R$ redrawn as a planar graph.}
	\label{Rplanar}
\end{figure}
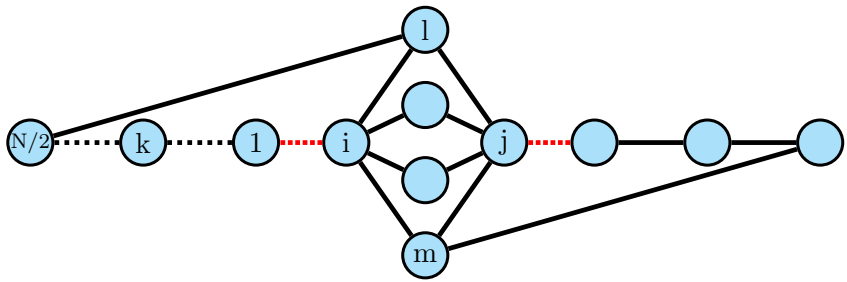

Because $G \cup R$ is planar, according to \cite{planar}, it is possible to determine the effective resistance for each node pair by a sequence of series, parallel, $Y-\Delta$, and $\Delta-Y$ transformations.

To analyse Eq.\ \eqref{eq_R_delta} it is convenient to consider the following two subsets of nodes in the graph $G \cup R$. First, denote the nodes in $K_{2,N-2}$ attached to the two path graphs as nodes $l$ and $m$, see Fig.\ \ref{fig_set_R}. Then we define the node set $A$ as the nodes in $K_{2,N-2}$ minus the nodes $i$, $j$, $l$ and $m$. Hence, $A$ consists of $N - 4$ nodes. Next, the subset of nodes $B$ is formed by the nodes in the two path graphs plus the nodes $l$ and $m$. Therefore, set $B$ contains $N + 2$ nodes. We can now split up the sum in Eq.\ \eqref{eq_R_delta} into three parts; (i) terms containing only $\omega_{ij}$ (2x), (ii) terms containing only nodes in set $A$ (iii) terms containing only nodes in set $B$. Then Eq.\ (\ref{eq_R_delta}) becomes
\begin{equation*}
	\Delta R_G = \frac{2N}{4(1+\omega_{ij})} \left( 2 \omega_{ij}^2 + \sum_{k\in A} (\omega_{ik}-\omega_{jk})^2 + \sum_{k\in B} (\omega_{ik}-\omega_{jk})^2 \right)
\end{equation*}
In the complete bipartite graph $K_{2,N-2}$, due to symmetry, it holds that $\omega_{ik} = \omega_{jk}$, so the contribution of set $A$ is zero. Thus, we can simplify $\Delta R_G$ to
\begin{equation}\label{eq_bsfysdifysif}
	\Delta R_G = \frac{2N}{4(1+\omega_{ij})} \left( 2 \omega_{ij}^2 + \sum_{k\in B} (\omega_{ik}-\omega_{jk})^2 \right)
\end{equation}
Hence, we need to compute $\omega_{ij}$ and $\omega_{ik}$ and $\omega_{jk}$ for nodes $k$ in the path graph. The index $k$ runs from $k=1$ to $k=N/2+1$, where the index $k=N/2+1$ corresponds to node $l$, see Fig.\ \ref{fig_set_R}.

As a first simplification step, we apply the resistors in series transformation twice, to the left part of the graph $G \cup R$ in Fig.\ \ref{fig_set_R}. We replace the path $i,1,\cdots,k$ by a single resistor of value $k$ and the path $i,l,N/2,\cdots,k$ by a single resistor of value $(\frac{N}{2} - k + 1)$. The result is depicted in Fig.\ \ref{fig_set_R2}. 

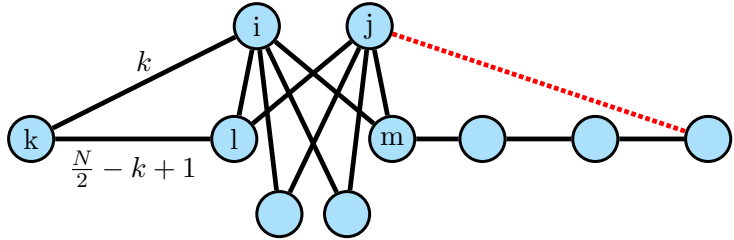
\begin{figure}[H]
	\centering
	\begin{tikzpicture}
		\node[mystyle, label=center:i] (0) at (0,1.5) {};
		\node[mystyle, label=center:j] (1) at (1.5,1.5) {};
		\node[mystyle, label=center:m] (2) at (1.8,0) {};
		\node[mystyle, label=center:l] (3) at (-0.3,0) {};
		\node[mystyle, label=center:k] (8) at (-3,0) {};
		\node[mystyle, label=center:] (6) at (3,0) {};
		\node[mystyle, label=center:] (7) at (4.5,0) {};
		\node[mystyle, label=center:] (10) at (6,0) {};
		\node[mystyle, label=center:] (11) at (0.3,-1) {};
		\node[mystyle, label=center:] (12) at (1.2,-1) {};
		
		\path [-,line width=1.8pt] (0) edge node[left] {} (2);
		\path [-,line width=1.8pt] (0) edge node[left] {} (3);
		\path [-,line width=1.8pt] (1) edge node[left] {} (2);
		\path [-,line width=1.8pt] (1) edge node[right] {} (3);
		\path [-,line width=1.8pt] (2) edge node[left] {} (6);
		\path [-,line width=1.8pt] (6) edge node[left] {} (7);
		\path [-,line width=1.8pt] (7) edge node[left] {} (10);
		\path [-,line width=1.8pt] (0) edge node[left] {} (11);
		\path [-,line width=1.8pt] (0) edge node[left] {} (12);
		\path [-,line width=1.8pt] (1) edge node[left] {} (11);
		\path [-,line width=1.8pt] (1) edge node[left] {} (12);
		\path [-,line width=1.8pt] (8) edge node[above] {$k$} (0);
		\path [-,line width=1.8pt] (8) edge node[below] {$\frac{N}{2}-k+1$} (3);
		
		\path [densely dotted,line width=1.8pt,color=red] (1) edge node[left] {} (10);
	\end{tikzpicture}
	\caption{Result of simplification step 1 on the graph $G \cup R$.}
	\label{fig_set_R2}
\end{figure}
Next we apply the series transformation to the path from node $j$ to $m$ on the right part of the graph, to obtain a link with a resistance of $\frac{N}{2}+1$. This results in two links between nodes $j$ and $m$. Applying the transformation for parallel resistors to these two links, we obtain a single link between nodes $j$ and $m$, with the following resistance:
\begin{equation*}
	r_{jm} = \frac{1}{ 1 + \frac{1}{\frac{N}{2} + 1}} = \frac{N+2}{N+4}
\end{equation*}
The resulting graph is depicted in Fig.\ \ref{fig_set_R3}. 
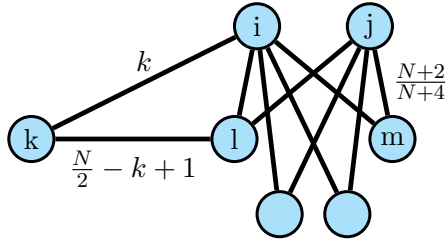
\begin{figure}[H]
	\centering
	\begin{tikzpicture}
		\node[mystyle, label=center:i] (0) at (0,1.5) {};
		\node[mystyle, label=center:j] (1) at (1.5,1.5) {};
		\node[mystyle, label=center:m] (2) at (1.8,0) {};
		\node[mystyle, label=center:l] (3) at (-0.3,0) {};
		\node[mystyle, label=center:k] (8) at (-3,0) {};
		\node[mystyle, label=center:] (11) at (0.3,-1) {};
		\node[mystyle, label=center:] (12) at (1.2,-1) {};
		
		\path [-,line width=1.8pt] (0) edge node[left] {} (2);
		\path [-,line width=1.8pt] (0) edge node[left] {} (3);
		\path [-,line width=1.8pt] (1) edge node[right] {$\frac{N+2}{N+4}$} (2);
		\path [-,line width=1.8pt] (1) edge node[right] {} (3);
		\path [-,line width=1.8pt] (0) edge node[left] {} (11);
		\path [-,line width=1.8pt] (0) edge node[left] {} (12);
		\path [-,line width=1.8pt] (1) edge node[left] {} (11);
		\path [-,line width=1.8pt] (1) edge node[left] {} (12);
		
		\path [-,line width=1.8pt] (8) edge node[above] {$k$} (0);
		\path [-,line width=1.8pt] (8) edge node[below] {$\frac{N}{2}-k+1$} (3);
	\end{tikzpicture}
	\caption{Result of simplification step 2 on the graph $G \cup R$.}
	\label{fig_set_R3}
\end{figure}
In Fig.\ \ref{fig_set_R3} there are $N - 2$ two-hop paths between nodes $i$ and $j$. Excluding the two paths passing through nodes $l$ and $m$, we can use the parallel resistors transformation to transform the remaining $N-4$ two-hop paths to a single link between nodes $i$ and $j$, with resistance $r_{ij}$ given by 
\begin{equation*}
	r_{ij} = \frac{1}{\frac{1}{2} (N-4)} = \frac{2}{N-4}
\end{equation*}
The resulting graph is depicted in Fig.\ \ref{fig_set_R4}
\begin{figure}[H]
	\centering
	\begin{tikzpicture}
		\node[mystyle, label=center:i] (0) at (0,1.5) {};
		\node[mystyle, label=center:j] (1) at (1.5,1.5) {};
		\node[mystyle, label=center:m] (2) at (1.8,0) {};
		\node[mystyle, label=center:l] (3) at (-0.3,0) {};
		\node[mystyle, label=center:k] (8) at (-3,0) {};
		
		\path [-,line width=1.8pt] (0) edge node[below,pos=0.6] {1} (2);
		\path [-,line width=1.8pt] (0) edge node[left] {1} (3);
		\path [-,line width=1.8pt] (1) edge node[right] {$\frac{N+2}{N+4}$} (2);
		\path [-,line width=1.8pt] (1) edge node[above,pos=0.7] {1} (3);
		
		\path [-,line width=1.8pt] (8) edge node[above] {$k$} (0);
		\path [-,line width=1.8pt] (8) edge node[below] {$\frac{N}{2}-k+1$} (3);
		
		\path [-,line width=1.8pt] (0) edge node[above] {$\frac{2}{N-4}$} (1);
	\end{tikzpicture}
	\caption{Result of simplification step 3 on the graph $G \cup R$.}
	\label{fig_set_R4}
\end{figure}
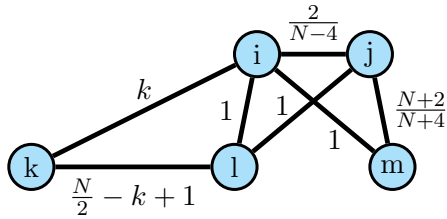
Next, we replace the path $i,m,j$ by a single link of resistance $1+\frac{N+2}{N+4}$. Then, again using the parallel resistors transformation between nodes $i$ and $j$, the resistance $r_{ij}'$ on the link $(i,j)$ becomes
\begin{equation*}
	r_{ij}' = \dfrac{1}{\dfrac{1}{\dfrac{2}{N-4}} + \dfrac{1}{1+\dfrac{N+2}{N+4}}} = \frac{2N+6}{N^2-8}
\end{equation*}
The resulting graph is depicted in Fig.\ \ref{fig_set_R5}.
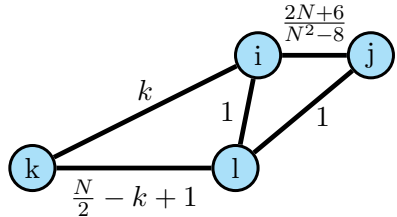
\begin{figure}[H]
	\centering
	\begin{tikzpicture}
		\node[mystyle, label=center:i] (0) at (0,1.5) {};
		\node[mystyle, label=center:j] (1) at (1.5,1.5) {};
		\node[mystyle, label=center:l] (3) at (-0.3,0) {};
		\node[mystyle, label=center:k] (8) at (-3,0) {};
		
		\path [-,line width=1.8pt] (0) edge node[above] {$\frac{2N+6}{N^2-8}$} (1);
		\path [-,line width=1.8pt] (0) edge node[left] {1} (3);
		\path [-,line width=1.8pt] (1) edge node[below,pos=0.3] {1} (3);
		
		\path [-,line width=1.8pt] (8) edge node[above] {$k$} (0);
		\path [-,line width=1.8pt] (8) edge node[below] {$\frac{N}{2}-k+1$} (3);
	\end{tikzpicture}
	\caption{Result of simplification step 4 on the graph $G \cup R$}
	\label{fig_set_R5}
\end{figure}

With the graph shown in Fig.\ \ref{fig_set_R5} we will now compute the resistances $\omega_{ij}, \omega_{ik}$ and $\omega_{jk}$. \\

\textbf{Resistance $\omega_{ij}$ between node $i$ and $j$:} \\
The links $(i,k)$ and $(k,l)$ in series, are parallel to the link $(i,l)$, such that we can replace the three links by a single link with resistance $r_{il}$, given by
\begin{equation*}
	r_{il} = \frac{N+2}{N+4}
\end{equation*}
The resulting graph is depicted in Fig.\ \ref{fig_set_R6a}

\begin{figure}[H]
	\centering
	\begin{tikzpicture}
		\node[mystyle, label=center:i] (0) at (0,1.5) {};
		\node[mystyle, label=center:j] (1) at (1.5,1.5) {};
		\node[mystyle, label=center:l] (3) at (-0.3,0) {};
		
		\path [-,line width=1.8pt] (0) edge node[above] {$\frac{2N+6}{N^2-8}$} (1);
		\path [-,line width=1.8pt] (0) edge node[left] {$\frac{N+2}{N+4}$} (3);
		\path [-,line width=1.8pt] (1) edge node[below,pos=0.3] {1} (3);
		
	\end{tikzpicture}
	\caption{Result of simplification step 5a on the graph $G \cup R$}
	\label{fig_set_R6a}
\end{figure}
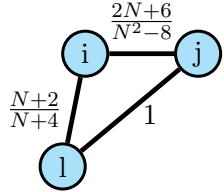

In the graph shown in Fig.\ \ref{fig_set_R6a} the links $(i,l)$ and $(l,j)$ in series, are parallel to the link $(i,j)$, such that we can replace the three links by a single link so we finally obtain the resistance $\omega_{ij}$:

\begin{equation}\label{om_ij}
	\omega_{ij} = \frac{2N+6}{N^2+N-4}
\end{equation}
\phantom{a}\\

\textbf{Resistance $\omega_{ik}$ between node $i$ and $k$:} \\
In Fig.\ \ref{fig_set_R5} the links $(i,j)$ and $(j,l)$ in series, are parallel to the link $(i,l)$, such that we can replace the three links by a single link with resistance $r_{il}'$, given by

\begin{equation*}
	r_{il}' = \frac{N^2+2N-2}{2N^2+2N-10}
\end{equation*}

The resulting graph is shown in Fig.\ \ref{fig_set_R6b}. 

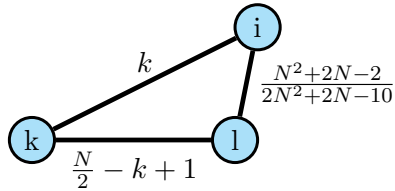
\begin{figure}[H]
	\centering
	\begin{tikzpicture}
		\node[mystyle, label=center:i] (0) at (0,1.5) {};
		\node[mystyle, label=center:l] (3) at (-0.3,0) {};
		\node[mystyle, label=center:k] (8) at (-3,0) {};
		
		\path [-,line width=1.8pt] (0) edge node[right] {$\frac{N^2+2N-2}{2N^2+2N-10}$} (3);
		
		\path [-,line width=1.8pt] (8) edge node[above] {$k$} (0);
		\path [-,line width=1.8pt] (8) edge node[below] {$\frac{N}{2}-k+1$} (3);
	\end{tikzpicture}
	\caption{Result of simplification step 5b on the graph $G \cup R$}
	\label{fig_set_R6b}
\end{figure}

In the graph shown in Fig.\ \ref{fig_set_R6b} the links $(i,l)$ and $(l,k)$ in series, are parallel to the link $(i,k)$, such that we can replace the three links by a single link so we finally obtain the resistance $\omega_{ik}$:

\begin{equation}\label{om_ik}
	\omega_{ik} = \frac{k (-2 k N^2 - 2 k N + 10 k + N^3 + 4 N^2 - N - 12)}{(N + 3) (N^2 + N - 4)}
\end{equation}
\phantom{a} \\

\textbf{Resistance $\omega_{jk}$ between node $j$ and $k$:} \\
For this case, the series and parallel transformations do not work. Instead, we use the $\Delta-Y$ transform on the right-hand side triangle in the graph in Fig.\ \ref{fig_set_R5}.
This comes down to removing the links $(i,j), (i,l)$ and $(j,l)$, adding a new node $A$ to the graph, and adding links $(i,A), (j,A)$ and $(l,A)$, with resistance $\frac{a}{2+a}$, $\frac{a}{2+a}$, and $\frac{1}{2+a}$, respectively, where $a=\frac{2N+6}{N^2-8}$. The resulting graph is depicted in Fig.\ \ref{fig_set_R6c}.

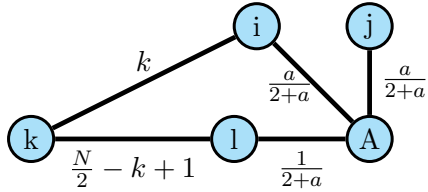
\begin{figure}[H]
	\centering
	\begin{tikzpicture}
		\node[mystyle, label=center:i] (0) at (0,1.5) {};
		\node[mystyle, label=center:j] (1) at (1.5,1.5) {};
		\node[mystyle, label=center:l] (3) at (-0.3,0) {};
		\node[mystyle, label=center:k] (8) at (-3,0) {};
		\node[mystyle, label=center:A] (7) at (1.5,0) {};
		
		
		\path [-,line width=1.8pt] (0) edge node[below,pos=0.2] {$\frac{a}{2+a}$} (7);
		\path [-,line width=1.8pt] (1) edge node[right] {$\frac{a}{2+a}$} (7);
		\path [-,line width=1.8pt] (3) edge node[below] {$\frac{1}{2+a}$} (7);
		
		\path [-,line width=1.8pt] (8) edge node[above] {$k$} (0);
		\path [-,line width=1.8pt] (8) edge node[below] {$\frac{N}{2}-k+1$} (3);
	\end{tikzpicture}
	\caption{Result of simplification step 5c on the graph $G \cup R$.}
	\label{fig_set_R6c}
\end{figure}
Now series and parallel transformations can be applied to obtain
\begin{equation*}
	\omega_{jk} = - \frac{18 + 12N + 2N^2 - 24k - 5Nk + 4N^2 k + N^3 k + 10 k^2 - 2 N k^2 - 2 N^2 k^2}{12 + N - 4 N^2 - N^3}
\end{equation*}
Having established $\omega_{ij}, \omega_{ik}$ and $\omega_{jk}$, we compute
\begin{equation*}
	\omega_{ik}-\omega_{jk} = \frac{2(2k - N - 3)}{N^2 + N - 4}
\end{equation*}
where the index $1 \leq k \leq \frac{N}{2}+1$. Then we find
\begin{equation}\label{eq_diff_om_ik_jk}
	\sum_{k=1}^{N/2+1} \left( \omega_{ik}-\omega_{jk} \right)^2 = \frac{2 (N^3 + 6N^2 + 11N + 6)}{3 (N^2 + N -4)^2}
\end{equation}
Substituting Eqs.\ (\ref{om_ij})--(\ref{eq_diff_om_ik_jk}) into Eq.\ (\ref{eq_bsfysdifysif}) finally gives
\begin{equation}\label{eq_deltaR}
	\Delta R_G = \frac{2N (N^3 + 12 N^2 + 47 N + 60)}{3 (N^4 + 4 N^3 + N^2 - 10 N - 8 )} =  \frac{2 \,(N+3)\,(N+4)\,(N+5)}{3\,(N+1)\,(N+2)\,(N^2+N-4)}
\end{equation}
Then the relative difference for the graph $G$ provided in Eq.\ (\ref{eq_deltaS}) and for the graph $G \cup R$ in Eq.\ (\ref{eq_deltaR}) finally yields
\begin{equation*}
	\gamma \leq  \frac{6\,(N+1)\,(N+2)\,(N^2+N-4)}{(N-2)\, N\, (N+3)\,(N+4)\,(N+5)}
\end{equation*}
For large $N$, we see that $\gamma \sim \frac{6}{N}$, which converges to zero.

\end{document}